\documentclass[english]{article}
\usepackage[T1]{fontenc}
\usepackage[latin9]{inputenc}
\usepackage{geometry}
\usepackage{babel}
\usepackage{units}
\usepackage{amsmath}
\usepackage{amsthm}
\usepackage{amssymb}
\usepackage[unicode=true,pdfusetitle,
 bookmarks=true,bookmarksnumbered=false,bookmarksopen=false,
 breaklinks=false,pdfborder={0 0 1},backref=false,colorlinks=false]
 {hyperref}

\makeatletter
\newcommand{\lyxaddress}[1]{
	\par {\raggedright #1
	\vspace{1.4em}
	\noindent\par}
}
\theoremstyle{plain}
\newtheorem{thm}{\protect\theoremname}

\usepackage{amsthm}
\usepackage{amsfonts}
\usepackage{graphicx}
\usepackage{float}
\usepackage{bbm}

\usepackage{tikz}

\let\originalleft\left
\let\originalright\right
\renewcommand{\left}{\mathopen{}\mathclose\bgroup\originalleft}
\renewcommand{\right}{\aftergroup\egroup\originalright}

\usepackage{amsmath,amssymb}

\usetikzlibrary{calc}
\usetikzlibrary{decorations.pathreplacing}

\usetikzlibrary{shapes}
\usetikzlibrary{arrows}
\usetikzlibrary{decorations}
\usetikzlibrary{decorations.pathmorphing}
\usetikzlibrary{patterns}
\usetikzlibrary{positioning}

\DeclareMathOperator{\claw}{\textsc{Claw}}
\DeclareMathOperator{\ed}{\textsc{Element-Distinctness}}
\DeclareMathOperator{\psearch}{\textsc{pSearch}}

\DeclareMathOperator{\E}{\mathbb{E}}

\newcommand{\cl}[2]{\claw_{#1\rightarrow #2}}

\date{}

\makeatother

\providecommand{\theoremname}{Theorem}

\begin{document}
\title{A Note About Claw Function With a Small Range\thanks{Supported by the project ``Quantum algorithms: from complexity theory
to experiment'' funded under ERDF programme 1.1.1.5.}}
\author{Andris Ambainis$^{1}$\and Kaspars Balodis$^{1}$\and J\={a}nis
Iraids$^{1}$}
\maketitle

\lyxaddress{$^{1}$ Center for Quantum Computer Science, Faculty of Computing,
University of Latvia}
\begin{abstract}
In the claw detection problem we are given two functions $f:D\rightarrow R$
and $g:D\rightarrow R$ ($|D|=n$, $|R|=k$), and we have to determine
if there is exist $x,y\in D$ such that $f(x)=g(y)$. We show that
the quantum query complexity of this problem is between $\Omega\left(n^{1/2}k^{1/6}\right)$
and $O\left(n^{1/2+\varepsilon}k^{1/4}\right)$ when $2\leq k<n$.
\end{abstract}

\section{Introduction}

In this note we study the $\claw$ problem in which given two discrete
functions $f:D\rightarrow R$ and $g:D\rightarrow R$ ($|D|=n$, $|R|=k$)
we have to determine if there is a collision, i.e., inputs $x,y\in D$
such that $f(x)=g(y)$. In contrast to the $\ed$ problem, where the
input is a single function $f:D\rightarrow R$ and we have to determine
if $f$ is injective, $\claw$ is non-trivial even when $k<n$. This
is the setting we focus on.

Both $\claw$ and $\ed$ have wide applications as useful subroutines
in more complex algorithms \cite{BJLM13,GS20} and as a means of lower
bounding complexity \cite{CKK12,ACL+20}.

$\claw$ and $\ed$ were first tackled by Buhrman et al. in 2000 \cite{BDH+05}
where they gave an $O\left(n^{3/4}\right)$ algorithm and $\Omega\left(n^{1/2}\right)$
lower bound. In 2003 Ambainis, introducing a novel technique of quantum
walks, improved the upper bound to $O\left(n^{2/3}\right)$ in the
query model \cite{ambainis2007quantum}. It was soon realized that
a similar approach works for $\claw$ \cite{CE05,MSS07,Tan09}. Meanwhile
Aaronson and Shi showed a lower bound $\Omega\left(n^{2/3}\right)$
that holds if the range $k=\Omega\left(n^{2}\right)$ \cite{aaronson2004quantum}.
Eventually Ambainis showed that the $\Omega\left(n^{2/3}\right)$
bound holds even if $k=n$ \cite{ambainis2005polynomial}. The same
lower bound has since been reproved using the adversary method \cite{Ros14}.
Until now, only the $\Omega\left(n^{1/2}\right)$ bound based on reduction
of searching was known for $\claw$ with $k=o\left(n\right)$ \cite{BDH+05}.

We consider quantum query complexity of $\claw$ where the input functions
are given as a list of their values in black box. Let $Q\left(f\right)$
denote the bounded error quantum query complexity of $f$. For a short
overview of black box model refer to Buhrman and de Wolf's survey
\cite{BdW02}. Let $[n]$ denote $\left\{ 1,2,\dots,n\right\} $.
Let $\cl{n}{k}:\left[k\right]^{2n}\rightarrow\left\{ 0,1\right\} $
be defined as

\[
\cl{n}{k}\left(x_{1},\dots,x_{n},y_{1},\dots,y_{n}\right)=\begin{cases}
1, & \text{if \ensuremath{\exists i,j\,x_{i}=y_{j}}}\\
0, & \text{otherwise}
\end{cases}.
\]

Our contribution is a quantum algorithm for $\cl{n}{k}$ showing $Q\left(\cl{n}{k}\right)=O\left(n^{1/2+\varepsilon}k^{1/4}\right)$
and a lower bound $Q\left(\cl{n}{k}\right)=\Omega\left(n^{1/2}k^{1/6}\right)$.
In section \ref{sec:Results} we describe the algorithm, and in section
\ref{sec:Lower-Bound} we give the lower bound.

\section{Results\label{sec:Results}}
\begin{thm}
\label{thm:alg}For all \textup{$\varepsilon>0$}, we have $Q\left(\cl{n}{k}\right)=O\left(n^{1/2+\varepsilon}k^{1/4}\right).$
\end{thm}
\begin{proof}
Let $X=\left(x_{1},\dots,x_{n}\right)$, $Y=\left(y_{1},\dots,y_{n}\right)$
be the inputs of the function. We denote $k=n^{\varkappa}$.

Consider the following algorithm parametrized by $\alpha\in\left[0,1\right]$.
\begin{enumerate}
\item Select a random sample $A=\left\{ a_{1},\dots,a_{\ell}\right\} \subseteq\left[n\right]$
of size $\ell=4\cdot n^{\alpha}\cdot\ln n$ and query the variables
$x_{a_{1}},\dots,x_{a_{\ell}}$.\\
Denote by $X_{A}=\left\{ x_{a}\mid a\in A\right\} $ the set containing
their values. Do a Grover search for an element $y\in Y$ such that
$y\in X_{A}$. If found, output 1.
\item[1'] Select a random sample $A'=\left\{ a'_{1},\dots,a'_{\ell}\right\} \subseteq Y$
of size $\ell$ and query the variables $y_{a'_{1}},\dots,y_{a'_{\ell}}$.\\
Denote by $Y_{A'}=\left\{ y_{a'}\mid a'\in A'\right\} $ the set containing
their values. Do a Grover search for an element $x\in X$ such that
$x\in Y_{A'}$. If found, output 1.
\item \label{enu:recstep}Run $\cl{4b\ln n}{k}$ algorithm (with the value
of $b$ specified below) with the following oracle:
\begin{enumerate}
\item To get $x_{i}$: do a pseudorandom permutation on $x_{1},\dots,x_{n}$
using seed $i$ and using Grover's minimum search return the first
value $x_{j}$ such that $x_{j}\notin X_{A}$.
\item To get $y_{i}$: do a pseudorandom permutation on $y_{1},\dots,y_{n}$
using seed $i$ and using Grover's minimum search return the first
value $y_{j}$ such that $y_{j}\notin X_{A'}$.
\end{enumerate}
\end{enumerate}
Let $B=\left\{ i\in\left[n\right]\mid x_{i}\notin X_{A}\right\} $,
$B'=\left\{ i\in\left[n\right]\mid y_{i}\notin Y_{A'}\right\} $ be
the sets containing the indices of the variables which have values
not seen in the steps 1 and 1'. We denote $\left|B\right|=b=n^{\beta}$.

Let us calculate the probability that after step 1 there exists an
unseen value $v$ which is represented in at least $n^{1-\alpha}$
variables, i.e., $v\notin X_{A}\wedge\left|\left\{ i\in\left[n\right]\mid x_{i}=v\right\} \right|\geq n^{1-\alpha}$.
Consider an arbitrary value $v^{*}\in\left[k\right]$ such that $\left|\left\{ i\mid x_{i}=v^{*}\right\} \right|\geq n^{1-\alpha}$.
For $i\in\left[\ell\right]$, let $Z_{i}$ be the event that $x_{a_{i}}=v^{*}$.
$\forall i\in\left[\ell\right]\ \Pr\left[Z_{i}\right]\geq\frac{n^{1-\alpha}}{n}$.
Let $Z=\sum_{i\in\left[\ell\right]}Z_{i}$. Then $\E\left[Z\right]=\ell\cdot\E\left[Z_{1}\right]\geq4\cdot n^{\alpha}\cdot\ln n\cdot\frac{n^{1-\alpha}}{n}=4\ln n$.
Using Chernoff inequality (see e.g. \cite{chung2006concentration}),
\[
\Pr\left[Z=0\right]\leq\exp\left(-\frac{1}{2}\E\left[Z\right]\right)\leq\exp\left(-2\ln n\right)=\frac{1}{n^{2}}.
\]
The probability that there exists such $v^{*}\in\left[k\right]$ is
at most $\frac{n^{\varkappa}}{n^{2}}=o\left(1\right)$. Therefore,
with probability $1-o\left(1\right)$ after step $1$, every value
$v\in B$ is represented in the input less than $n^{1-\alpha}$ times.
The same reasoning can be applied to step $1'$ and the set $B'$.
Therefore, with probability $1-o\left(1\right)$ both $b$ and $b'$
are at most $k\cdot n^{1-\alpha}=n^{\varkappa+1-\alpha}$.

Similarly, we show that with probability $1-o\left(1\right)$ each
$x\in B$ appears as the first element from $B$ in at least one of
the permutations of the oracle in step 2. Let $W_{i}^{x}$ be the
event that $x\in B$ appears in the $i$-th permutation as the first
element from $B$. $\E\left[W_{i}^{x}\right]=\frac{1}{b}$. Let $W^{x}=\sum_{i\in\left[4b\ln n\right]}W_{i}^{x}$.
$\E\left[W^{x}\right]=4b\ln n\cdot\frac{1}{b}=4\ln n$. $\Pr\left[W^{x}=0\right]\leq\exp\left(-2\ln n\right)=\frac{1}{n^{2}}$.
$\Pr\left[\exists x\in B:Z^{x}=0\right]\leq\frac{n}{n^{2}}=\frac{1}{n}=o\left(1\right)$.
The same argument works for $B'$. Therefore, if there is a collision,
it will be found by the algorithm with probability $1-o\left(1\right)$.

We also show that with probability $1-o\left(1\right)$, in all permutations
the first element from $B$ appears no further than in position $4\frac{n}{b}\ln n$
(and similarly for $B'$). We denote by $P_{i,j}$ the event that
in the $i$-th permutation in the $j$-th position is an element from
$B$. $\E\left[P_{i,j}\right]=\frac{b}{n}$. We denote $P_{i}=\sum_{j\in\left[4\cdot\frac{n}{b}\cdot\ln n\right]}P_{i,j}$.
$\E\left[P_{i}\right]=4\cdot\ln n$. $\Pr\left[P_{i}=0\right]\leq\exp\left(-2\ln n\right)=\frac{1}{n^{2}}$.
$\Pr\left[\exists i\in\left[4b\ln n\right]:P_{i}=0\right]\leq\frac{4b\ln n}{n^{2}}\leq\frac{4n\ln n}{n^{2}}=o\left(1\right)$.
Therefore, the Grover's minimum search will use at most $\tilde{O}\left(\sqrt{\frac{n}{n^{\beta}}}\right)$
queries.

The steps 1 and 1' use $\tilde{O}\left(n^{\alpha}\right)$ queries
to obtain the random sample, and $O\left(\sqrt{n}\right)$ queries
to check if there is a colliding element on the other side of the
input. The oracle in step 2 uses $\tilde{O}\left(\sqrt{\frac{n}{n^{\beta}}}\right)$
queries to obtain one value of $x_{i}$ or $y_{i}$.

Therefore the total complexity of the algorithm is 
\[
\tilde{O}\left(n^{\alpha}+n^{\frac{1}{2}}+Q\left(\cl{4b\ln n}{k}\right)\cdot n^{\frac{1}{2}-\frac{1}{2}\beta}\right).
\]

By using the $O\left(n^{2/3}\right)$ algorithm in step 2, 
\begin{align*}
Q\left(\cl{4b\ln n}{k}\right)\cdot n^{\frac{1}{2}-\frac{1}{2}\beta} & =n^{\frac{2}{3}\beta+\frac{1}{2}-\frac{1}{2}\beta}\\
 & =n^{\frac{1}{2}+\frac{1}{6}\beta}\\
 & \leq n^{\frac{1}{2}+\frac{1}{6}\left(\varkappa+1-\alpha\right)}\\
 & =n^{\frac{4+\varkappa-\alpha}{6}},
\end{align*}

and the total complexity is minimized by setting $\alpha=\frac{4+\varkappa}{7}$.
However, we can do better than that. Notice that the $O\left(n^{2/3}\right)$
algorithm might not be the best choice for solving $\cl{4b\ln n}{k}$
in step 2.

Let $\mathcal{A}_{0}$ denote the regular $O\left(n^{\nicefrac{2}{3}}\right)$
$\cl{n}{k}$ algorithm. For $i>0$, let $\mathcal{A}_{i}$ denote
a version of algorithm from Theorem \ref{thm:alg} that in step \ref{enu:recstep}
calls $\mathcal{A}_{i-1}$. Then we show that for all $n$ and all
$0\leq\varkappa\leq\frac{2}{3}$, 
\[
Q\left(\mathcal{A}_{i}\right)=\tilde{O}\left(n^{T_{i}(\varkappa)}\right),
\]
 where $T_{i}(\varkappa)=\frac{\left(2^{i}-1\right)\varkappa+2^{i+1}}{2^{i+2}-1}$.

The proof is by induction on $i$. For $i=0$, we trivially have that
$Q\left(\mathcal{A}_{0}\right)=\tilde{O}\left(n^{\nicefrac{2}{3}}\right)$.
For the inductive step, consider the analysis of our algorithm. Let
us set $\alpha=T_{i}\left(\varkappa\right)$. First, notice that $T_{i}\text{\ensuremath{\left(\varkappa\right)}}$
is non-decreasing in $\varkappa$ and $T_{i}\left(\frac{2}{3}\right)=\frac{2}{3}$
for all $i$. Thus for all $\varkappa\leq\frac{2}{3}$, we have $T_{i}\left(\varkappa\right)\leq\frac{2}{3}$,
hence $\alpha\leq\frac{2}{3}$ and $\frac{\varkappa}{1-\alpha+\varkappa}\leq\frac{2}{3}$.
Second, since the coefficient of $\varkappa$ is $\frac{2^{i}-1}{2^{i+2}-1}\leq1$
the function $T_{i}\left(\varkappa\right)$ is above $\varkappa$
for $\varkappa\leq\frac{2}{3}$, establishing $\alpha-\varkappa\geq0$.
This confirms that $\alpha=T_{i}\left(\varkappa\right)$ is a valid
choice of $\alpha$.

It remains to show that the complexity of step \ref{enu:recstep}
does not exceed $T_{i}\left(\varkappa\right)$. By the inductive assumption
and analysis of the algorithm, the complexity (up to logarithmic factors)
of the second step is $n$ to the power of $\left(1-\alpha+\varkappa\right)\cdot T_{i-1}\left(\frac{\varkappa}{1-\alpha+\varkappa}\right)+\frac{\alpha-\varkappa}{2}$.
Finally, we have to show that 
\[
\left(1-T_{i}\left(\varkappa\right)+\varkappa\right)\cdot T_{i-1}\left(\frac{\varkappa}{1-T_{i}\left(\varkappa\right)+\varkappa}\right)+\frac{T_{i}\left(\varkappa\right)-\varkappa}{2}\leq T_{i}\left(\varkappa\right).
\]
By expanding $T_{i-1}\left(\varkappa\right)$ and with a slight rearrangement,
we obtain 
\[
\frac{(2^{i-1}-1)\varkappa+2^{i}\left(1-T_{i}\left(\varkappa\right)+\varkappa\right)}{2^{i+1}-1}\leq\frac{T_{i}\left(\varkappa\right)+\varkappa}{2}.
\]
We can further rearrange the required inequality by bringing $T_{i}\left(\varkappa\right)$
to right hand side and everything else to the other. Then we get
\[
\frac{(2^{i-1}-1+2^{i}-\frac{2^{i+1}-1}{2})\varkappa+2^{i}}{2^{i+1}-1}\leq T_{i}\left(\varkappa\right)\left(\frac{1}{2}+\frac{2^{i}}{2^{i+1}-1}\right).
\]
After simplification we obtain $\frac{\left(2^{i}-1\right)\varkappa+2^{i+1}}{2^{i+2}-1}\leq T_{i}(\varkappa)$,
which is true. 

Since $\lim_{i\rightarrow\infty}\frac{2^{i}-1}{2^{i+2}-1}=\frac{1}{4}$
and $\lim_{i\rightarrow\infty}\frac{2^{i+1}}{2^{i+2}-1}=\frac{1}{2}$,
the result follows.
\end{proof}

\section{Lower Bound\label{sec:Lower-Bound}}

We show a $\Omega\left(n^{1/2}k^{1/6}\right)$ quantum query complexity
lower bound for $\cl{n}{k}$.
\begin{thm}
For all $k\geq2$, we have $Q\left(\cl{n}{k}\right)=\Omega\left(n^{1/2}k^{1/6}\right)$.
\end{thm}
\begin{proof}
Let $\psearch_{m}:\left({*}\cup[k]\right)^{m}\rightarrow[k]$ be the
partial function defined as 
\[
\psearch_{m}\left(x_{1},x_{2},\ldots,x_{m}\right)=\begin{cases}
x_{i}, & \text{if }x_{i}\neq*,\forall j\neq i:x_{j}=*\\
\text{undefined}, & \text{otherwise}
\end{cases}.
\]
Consider the function $f_{n,k}=\cl{k}{k}\circ\psearch_{\left\lfloor n/k\right\rfloor }$.
One can straightforwardly reduce $f_{n,k}(x,y)$ to $\cl{n}{k+2}(x',y')$
by setting 
\[
x'_{i}=\begin{cases}
x_{i}, & \text{if }x_{i}\neq*\\
k+1, & \text{if }x_{i}=*
\end{cases}
\]
 and 
\[
y'_{i}=\begin{cases}
y_{i}, & \text{if }y_{i}\neq*\\
k+2, & \text{if }y_{i}=*
\end{cases}.
\]

Next, we show that $Q\left(f_{n,k}\right)=\Omega\left(k^{2/3}\sqrt{n/k}\right)=\Omega\left(n^{1/2}k^{1/6}\right)$.
The fact that $Q\left(\cl{k}{k}\right)=\Omega\left(k^{2/3}\right)$
has been established by Zhang \cite{Zha05}. Furthermore, thanks to
the work done by Brassard et al. in \cite[Theorem 13]{BHK+19} we
know that for $\psearch_{m}$ a composition theorem holds: $Q\left(h\circ\psearch_{m}\right)=\Omega\left(Q\left(h\right)\cdot Q\left(\psearch_{m}\right)\right)=\Omega(Q(h)\cdot\sqrt{m})$.
Therefore, 
\[
Q\left(\cl{n}{k}\right)\geq Q\left(\cl{k-2}{k-2}\circ\psearch_{\left\lfloor \frac{n}{k-2}\right\rfloor }\right)=\Omega\left(k^{2/3}\sqrt{\frac{n}{k}}\right)=\Omega\left(n^{1/2}k^{1/6}\right).
\]
\end{proof}

\section{Open Problems}

Can we show that $Q\left(\cl{n}{n^{2/3}}\right)=\Omega\left(n^{\nicefrac{2}{3}}\right)$?
In particular, our algorithm struggles with instances where there
are $\frac{n^{\nicefrac{2}{3}}}{2}$ singletons only two (or none)
of which are matching and the remaining variables are evenly distributed
with $\Theta\left(n^{\nicefrac{1}{3}}\right)$ copies each, such that
none are matching. Thus our algorithm then either has to waste time
sampling all the high-frequency decoy values or have most variables
not sampled by step \ref{enu:recstep}. If this lower bound held,
it would imply a better lower bound for evaluating constant depth
formulas and Boolean matrix product verification \cite[Theorem 5]{CKK12}.

\bibliographystyle{plain}
\bibliography{claw-ref}

\begin{thebibliography}{10}

\bibitem{ACL+20}
Scott Aaronson, Nai-Hui Chia, Han-Hsuan Lin, Chunhao Wang, and Ruizhe Zhang.
\newblock {On the Quantum Complexity of Closest Pair and Related Problems}.
\newblock In Shubhangi Saraf, editor, {\em 35th Computational Complexity
  Conference (CCC 2020)}, volume 169 of {\em Leibniz International Proceedings
  in Informatics (LIPIcs)}, pages 16:1--16:43, Dagstuhl, Germany, 2020. Schloss
  Dagstuhl--Leibniz-Zentrum f{\"u}r Informatik.

\bibitem{aaronson2004quantum}
Scott Aaronson and Yaoyun Shi.
\newblock Quantum lower bounds for the collision and the element distinctness
  problems.
\newblock {\em Journal of the ACM (JACM)}, 51(4):595--605, 2004.

\bibitem{ambainis2005polynomial}
Andris Ambainis.
\newblock Polynomial degree and lower bounds in quantum complexity: Collision
  and element distinctness with small range.
\newblock {\em Theory of Computing}, 1(1):37--46, 2005.

\bibitem{ambainis2007quantum}
Andris Ambainis.
\newblock Quantum walk algorithm for element distinctness.
\newblock {\em SIAM Journal on Computing}, 37(1):210--239, 2007.

\bibitem{BJLM13}
Daniel~J. Bernstein, Stacey Jeffery, Tanja Lange, and Alexander Meurer.
\newblock Quantum algorithms for the subset-sum problem.
\newblock In Philippe Gaborit, editor, {\em Post-Quantum Cryptography}, pages
  16--33, Berlin, Heidelberg, 2013. Springer Berlin Heidelberg.

\bibitem{BHK+19}
Gilles Brassard, Peter H\o{}yer, Kassem Kalach, Marc Kaplan, Sophie Laplante,
  and Louis Salvail.
\newblock Key establishment \`{a} la merkle in a quantum world.
\newblock {\em Journal of Cryptology}, 32(3):601--634, July 2019.

\bibitem{BdW02}
Harry Buhrman and Ronald {de Wolf}.
\newblock Complexity measures and decision tree complexity: a survey.
\newblock {\em Theoretical Computer Science}, 288(1):21--43, 2002.
\newblock Complexity and Logic.

\bibitem{BDH+05}
Harry Buhrman, Christoph D\"{u}rr, Mark Heiligman, Peter H\o{}yer,
  Fr\'{e}d\'{e}ric Magniez, Miklos Santha, and Ronald de~Wolf.
\newblock Quantum algorithms for element distinctness.
\newblock {\em SIAM Journal on Computing}, 34(6):1324--1330, 2005.

\bibitem{CE05}
Andrew~M. Childs and Jason~M. Eisenberg.
\newblock Quantum algorithms for subset finding.
\newblock {\em Quantum Info. Comput.}, 5(7):593--604, November 2005.

\bibitem{CKK12}
Andrew~M. Childs, Shelby Kimmel, and Robin Kothari.
\newblock The quantum query complexity of read-many formulas.
\newblock In {\em Proceedings of the 20th Annual European Conference on
  Algorithms}, ESA'12, pages 337--348, Berlin, Heidelberg, 2012.
  Springer-Verlag.

\bibitem{chung2006concentration}
Fan Chung and Linyuan Lu.
\newblock Concentration inequalities and martingale inequalities: a survey.
\newblock {\em Internet Mathematics}, 3(1):79--127, 2006.

\bibitem{GS20}
Fran\c{c}ois~Le Gall and Saeed Seddighin.
\newblock Quantum meets fine-grained complexity: Sublinear time quantum
  algorithms for string problems.
\newblock 2020.

\bibitem{MSS07}
Fr\'{e}d\'{e}ric Magniez, Miklos Santha, and Mario Szegedy.
\newblock Quantum algorithms for the triangle problem.
\newblock {\em SIAM Journal on Computing}, 37(2):413--424, 2007.

\bibitem{Ros14}
Ansis Rosmanis.
\newblock Adversary lower bound for element distinctness with small range.
\newblock 2014.

\bibitem{Tan09}
Seiichiro Tani.
\newblock Claw finding algorithms using quantum walk.
\newblock {\em Theoretical Computer Science}, 410(50):5285--5297, 2009.
\newblock Mathematical Foundations of Computer Science (MFCS 2007).

\bibitem{Zha05}
Shengyu Zhang.
\newblock Promised and distributed quantum search.
\newblock In Lusheng Wang, editor, {\em Computing and Combinatorics}, pages
  430--439, Berlin, Heidelberg, 2005. Springer Berlin Heidelberg.

\end{thebibliography}

\end{document}